\documentclass[runningheads]{llncs}
\usepackage[T1]{fontenc}
\usepackage{graphicx}
\usepackage{amsmath}
\usepackage{amssymb}
\usepackage[plain,boxed,vlined,linesnumbered,inoutnumbered,rightnl]{algorithm2e}
\usepackage{algorithmic}
\usepackage{multirow}
\usepackage{cleveref}

\newcommand{\chainsum}{\mathsf{sum}}
\newcommand{\cost}{\mathsf{cost}}
\newcommand{\win}{\mathsf{win}}

\newcommand{\smax}{\mathsf{win}^*}
\newcommand{\nmax}{\mathsf{win}^{-}}
\newcommand{\next}{\mathsf{next}}
\newcommand{\Children}{\mathsf{Child}}



\begin{document}
\title{Sum-of-Max Chain Partition of a Tree}
%
\author{Ruixi Luo\inst{1} \and Taikun Zhu\inst{1} \and Kai Jin\inst{1}}
%
%
\institute{School of Intelligent Systems Engineering, Sun Yat-Sen University, No. 66, Gongchang Road, Guangming District, Shenzhen, Guangdong, China}

%
\maketitle              
%


\begin{abstract}
Path partition problems on trees have found various applications. 
In this paper, we present an $O(n \log n)$ time algorithm for solving the following variant of path partition problem: given a rooted tree of $n$ nodes $1, \ldots , n$, where vertex $i$ is associated with a weight $w_i$ and a cost $s_i$, partition the tree into several disjoint chains $C_1,\ldots,C_k$, so that the weight of each chain is no more than a threshold $w_0$ and the sum of the largest $s_i$ in each chain is minimized. We also generalize the algorithm to the case where the cost of a chain is determined by the $s_i$ of the vertex with the highest rank in the chain, which can be determined by an arbitrary total order defined on all nodes instead of the value of $s_i$.

\end{abstract}

\section{Introduction}

Sequence and tree partition problems have been extensively studied in the past decades~\cite{bagheri2024linear,bagheri2010minimum,boesch1974covering,pinar2004fast}. In this paper, we study a variant of the tree partition problems, namely, sum-of-max chain partition of a tree.

Above all, it is necessary to clarify the definition of a chain in the context of tree partitioning. We adopt the chain partition definition as in Misra and Tarjan's work~\cite{misra1975optimal}. A chain in this context is defined as follows:
A path of a rooted tree is a \emph{chain} if for each node pair $(u, v)$ on the path, either $u$ is the ancestor of $v$, or $v$ is the ancestor of $u$.

\begin{quote}\label{problem 1}
	\textbf{Sum-of-max chain partition problem.} 
	Given a rooted tree $T$ of $n$ nodes $1, \ldots , n$, where each vertex $i$ is associated with a weight $w_i$ and a cost $s_i$, partition the tree into several disjoint chains $C_1,\ldots,C_k$, where $k$ is arbitrary, so that the total weight of each chain is no more than a given threshold $w_0$ and the sum of the largest $s_i$ in each chain is minimized. Without loss of generality, we assume that $w_i \leq w_0$ for all $i$.
\end{quote}

Let the largest $s_i$ in each chain be referred to as the \textbf{cost} of the chain. The total cost of a chain partition $P = \bigcup_{i=1}^k C_i$ is the sum of the largest parameter $s$ in each chain, as defined above. Formally, $\cost(P)=\sum_{1\leq i\leq k} \max_{j \in C_i}s_j$. 


\medskip Partitioning a tree into paths has been studied by a number of researchers. In~\cite{yan1997k}, the problem of partitioning a tree into paths with each path having no more than $k$ nodes was solved with a linear time algorithm. The sum-of-max chain partition problem on trees can be seen as a generalization of the problem with variable weights and costs on nodes to some extent. The main contribution of this paper is an $O(n \log n)$ time algorithm for solving sum-of-max chain partition problem on trees mentioned above.

\subsection{Related works}

Misra and Tarjan~\cite{misra1975optimal} investigated a type of chain partition problem 
  where a tree is divided into several chains, each of which does not exceed the weight limit $w_0$, 
    and the objective is to minimize the cost of edges contained in no chains. 
  An $O(n \log n)$ algorithm for this problem was presented in their paper.

The \emph{heavy-light decomposition} technique ~\cite{sleator1981data} introduces a type of chain partition of trees, which finds numerous applications~\cite{alstrup1997finding,buchsbaum2000maintaining,harel1984fast,klein1998computing}. 
Similarly, the \emph{long-path decomposition} technique \cite{bender2004level} defines another chain partition of trees,
      which can be used to speed up several dynamic programming algorithms on trees.
   Both of the mentioned chain partitions can be computed in linear time.

\medskip 
Path partition problems have drawn more attention in the literature~\cite{brause2017relation,chen2019improved,li2024improved,monnot2007path}.
For example, partitioning a graph $G$ into the minimum number of paths has been studied extensively~\cite{boesch1974covering,lu2013path,manuel2018revisiting}. Boesch et al.~\cite{boesch1974covering} provided a simple polynomial algorithm for the special case where $G$ is a tree. This type of partition problem on trees is related to the Hamiltonian completion of a tree~\cite{goodman1974hamiltonian,slater1979path}.  

Some path partition problems of trees with costs related to the number of paths have also been studied in the past. 
Cai et al.~\cite{cai2003} studied $\omega$-path partition of a edge-weighted tree and gave a linear time algorithm for the problem. 
Jin and Li~\cite{jin2008vertex} presented a linear path partition algorithm for edge-colored trees.
Yan et al.~\cite{yan1997k} presented a linear-time algorithm for the $k$-path partition problem in trees,
which considers partitioning a tree into the smallest number of paths such that each path has at most $k$ vertices. 
The path partitions problems can be seen as special cases of sum-of-max path partition problem having uniform costs for all elements.


As a sequence can be treated as a tree without branches, the problem in this paper can also be seen as an extension of the sum-of-max partition problem under a knapsack constraint studied in~\cite{jin2023sum}, where an $O(n)$ time algorithm is given. 

Partition a tree into subtrees has also been investigated extensively. A few references about such problems can be found in the introduction of \cite{jin2023sum}.


\medskip The decision version of our chain partition problem is NP-complete for a directed graph $G$. 
  This is because the chain partition problem can be reduced to the minimum path cover problem when all weights and costs are set to be $1$. 
  Then, the minimum path cover problem can further be reduced to a classic NP-complete problem, namely the Hamiltonian path problem~\cite{hartmanis1982computers}, since a minimum path cover consists of one path if and only if there is a Hamiltonian path in $G$.

\subsection{Applications}

The partition problem on trees can be applied in real-life, such as water environment protection, industrial investment planning, power grid maintenance, biological sample supply and preservation, etc.

In a river environmental protection scenario, we need to send several staff members to clean the river. Each staff member is responsible for a continuous section of the river. Since the river forks, each unbranched river segment can be seen as a vertex of the tree.

Assume that each person can clean a river of length $w_0$ at most, each segment $i$ is associated with length $w_i$ and equipment requirement $s_i$. The equipment that each staff member needs is determined by the largest $s_i$ of the river segments he is responsible for. An effective sum-of-max chain partition guarantees a good staff allocation plan with minimum total equipment requirement.

Our sum-of-max chain partition problem can also be applied to industrial investment. Assume a cooperation focuses on oil related industries. Then, the supply chains based on oil can be seen as a tree, with oil as the root, various consumer goods as leaf vertices and intermediate products between them. Now, the cooperation intends to establish a number of factories, each produces several products on the tree.

In order to ensure the continuity of production, only continuous products on the tree, i.e. a chain, can serve as the supply chain segment for a factory's production. Suppose a product $i$ on the tree requires $w_i$ area of land, and $s_i$ level of pollution. Each factory can occupy no more than $w_0$ area of land and its pollution is determined by the largest $s_i$ of the products it produces. A good sum-of-max chain partition can reduce the total pollution of the factories.

\section{Preliminaries}\label{sect:pre}

Denote the given tree and its subtree rooted at vertex $v$ by $T$ and $T_v$, respectively.  
For each vertex $u$ in $T_v$, 
    denote by $T[v,u]$ the chain from $v$ to $u$,
       by $T[v,u)$ the chain from $v$ to the parent of $u$, and by $T(v,u]$ the chain $T[v,u]$ without $v$.
By this definition, $T[v,v)$ is empty. The set of children of $v$ is denoted by $\Children(v)$.

\begin{definition}\label{def:preliminary}
The \emph{window} of $v$, denoted by $\win_{v}$,
    refers to the set of vertices $u$ in $T_v$ for which the chain $T[v,u]$ weights no more than $w_0$.

A vertex $u$ in $T_v$ is \emph{s-maximal} if each vertex $p$ in $T[v,u)$ admits that $s_p<s_u$. 
For $u\in T_v$ and $u\neq v$, let $\next_v(u)$ be the closest proper s-maximal ancestor of $u$ in $T_v$.
Note that $\next_v(u)$ is well-defined since the root $v$ of $T_v$ is always s-maximal in $T_v$. 

Let $\smax_v$ denote the s-maximal vertices in $\win_v$ and $\nmax_v$ denote vertices in $\win_v\setminus\smax_v$.
\end{definition}

Our algorithm follows a dynamic programming paradigm. 

Let $F[v]$ denote the minimum cost to partition $T_v$ into chains.
Obviously, our goal is to compute $F[r]$, where $r$ denotes the root of $T$.

Let us consider how to find the optimal partition of $T_v$.
Obviously, $v$ must belong to some chain, e.g. $T[v,i]$.
Define $\cost(v,i)$ to be the minimum cost to partition $T_v$ into chains, one of which equals $T[v,i]$. Then, we have

\begin{equation}\label{eqn:F-cost}
	F[v]=\min_{\text{$T[v,i]$ weights no more than $w_0$}} \cost(v,i).
\end{equation}

\begin{figure}[h]
\begin{minipage}{0.5\textwidth}
	\centering
	\includegraphics[width=0.57\textwidth]{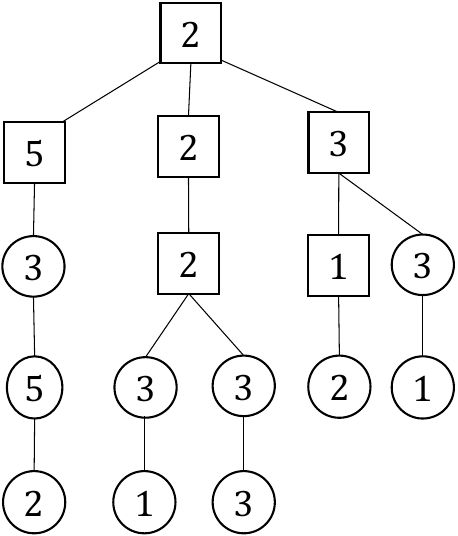} (a)

\end{minipage}
\begin{minipage}{0.5\textwidth}
	\includegraphics[width=0.66\textwidth]{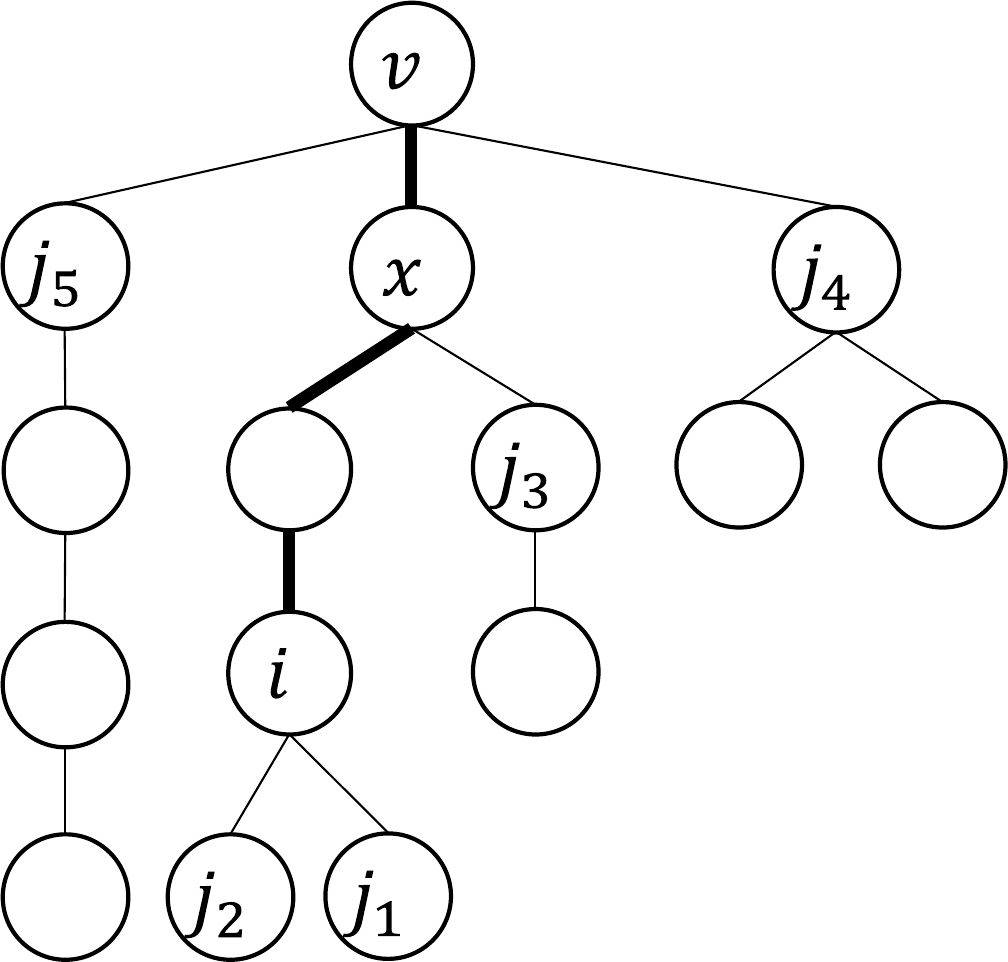} (b)
\end{minipage}
	\caption{(a) Illustration for window and outer boundary, where $w_0=7$. 
        The numbers in the vertices are their weights.
        Squares represent the vertices in the window.\\
	(b) Illustration for $\chainsum(v,i)$. Here, $\chainsum(v,i)= F[j_1]+F[j_2]+F[j_3]+F[j_4]+F[j_5]$.}\label{fig:pic-1}
\end{figure}

Recall Definition~\ref{def:preliminary}. We rewrite equation~(\ref{eqn:F-cost}) as:
\begin{equation}\label{eq:F_pre}
	F[v]=\min_{i\in \win_v}\cost(v,i).
\end{equation}

Recall $\next$ in Definition~\ref{def:preliminary}. A trivial formula of $\cost(v,i)$ is as follows: 
\begin{equation}\label{eq:cost}
	\cost(v,i)= 
	\begin{cases} 
		\chainsum(v,i)+s_{i} & \text{if } i \in \smax_v\\
		\chainsum(v,i)+s_{\next_v(i)} & \text{if } i \in \nmax_v,
	\end{cases}
\end{equation}
where
\begin{equation}\label{eq:chainsum}
	\chainsum(v,i):= \sum_{j\textnormal{'s parent} \in T[v,i], j\notin  T[v,i]}F[j].
\end{equation}

See Figure~\ref{fig:pic-1}~(b) for an illustration of $\chainsum(v,i)$.

By the definitions above, we can derive the following lemma, which is crucial for the design and analysis in this paper.

\begin{lemma}\label{lem:cost substract}
 \begin{enumerate}
  Assume $i\in \nmax_v$ and $x=\next_v(i)$. Then,
    \begin{equation}\label{eqn:cost_substract_1}
        i\in\nmax_x\text{ and }x=\next_x(i),
    \end{equation}
  and more importantly,
      \begin{equation}\label{eqn:cost_substract}
           \cost(v,i)-\cost(x,i)=\cost(v,x)-\cost(x,x).
       \end{equation}
  \end{enumerate}
\end{lemma} 
 
\begin{proof}
	We may assume $v\neq x$ (otherwise it is trivial).
	
	Since $i\in \nmax_v$ and $x=\next_v(i)$, we know that $s_x\geq s_p$ for $p\in T[v,i]$. 
    In particular, $s_x\geq s_p$ for $p\in T[x,i]$. Consequently, in $T[x,i]$, only $x$ is s-maximal in $T_x$ and thus 
    we obtain equation~(\ref{eqn:cost_substract_1}).
    
    Further applying equation~(\ref{eq:cost}), we get
         $$\left\{
             \begin{array}{ll}
               \cost(v,i)=\chainsum(v,i)+s_x \\
               \cost(x,i) = \chainsum(x,i)+s_x.
             \end{array}
           \right.
         $$

    Let $p_x$ denote the parent of $x$. 
	\begin{equation}\label{eq:cost-cost}
		\begin{aligned}
			\cost(v,i)-\cost(x,i)&= (\chainsum(v,i)+s_x)-(\chainsum(x,i)+s_x)\\
			&= \chainsum(v,i)-\chainsum(x,i) \\
            &= \chainsum(v,p_x)-F[x].
		\end{aligned}
	\end{equation}
    The last equation follows from the definition~(\ref{eq:chainsum}).

	\smallskip Similarly, we have	
	\begin{equation}\label{eq:cost-cost2}
		\begin{aligned}
			\cost(v,x)-\cost(x,x)&= (\chainsum(v,x)+s_x)-(\chainsum(x,x)+s_x)\\
            &= \chainsum(v,x)-\chainsum(x,x) \\
            &= \chainsum(v,p_x)-F[x].
		\end{aligned}
	\end{equation}
	
	Equations~(\ref{eq:cost-cost}) and (\ref{eq:cost-cost2}) together imply equation~(\ref{eqn:cost_substract}). \qed
\end{proof}

We will use a few heaps to compute the minimum cost $F[v]$. For simplicity and time complexity, all heaps in this paper refer to \emph{binomial heaps}~\cite{Vuillemin1978A}.
  (The \emph{leftist tree} \cite{mehta2004handbook} is not appropriate for our purpose.)

\newcommand{\findmin}{\textnormal{find-min}}
\newcommand{\findmax}{\textnormal{find-max}}
\newcommand{\ins}{\textnormal{insert}}
\newcommand{\deletevertex}{\textnormal{delete-vertex}}
\newcommand{\meld}{\textnormal{meld}}
\newcommand{\addall}{\textnormal{add-all}}
\newcommand{\keyvalue}{\textnormal{key-value}}

\begin{table}[h]
	\centering
	\begin{tabular}{||c c c||} 
		\hline
		Operation & Definition & Time required\\ [0.5ex] 
		\hline\hline
		find-min(max)($H$) & finds the minimal(maximal) element in heap $H$. & $O(1)$\\ 
		insert($H$,$i$,$k$) & inserts vertex $i$ in heap $H$ with key $k$. & $O(\log n)$\\
		delete-vertex($H$,$i$) & deletes vertex $i$ in heap $H$ (if exists). & $O(\log n)$\\
		meld($H_1$,$H_2$) & melds heap $H_1$ with $H_2$ into $H_1$. & $O(\log n)$\\
		add-all($H$,$v$) & adds value $v$ to all key values in $H$. & $O(\log n)$\\
		key-value($H$,$i$) & returns key value of $i$ in $H$ (if exists). & $O(\log n)$\\[1ex] 
		\hline
	\end{tabular}
	\vspace{2mm}
	\caption{Operations for the heaps.}
	\label{table:operation}
\end{table}

Operations shown in Table~\ref{table:operation} are needed for the binomial heaps in the algorithm. Most of the operations are originally available for the binomial heaps~\cite{Vuillemin1978A}. Additionally, we implement a lazy tag technique to create an add-all() operation in Table~\ref{table:operation}. See appendix~\ref{app:lazy_tag} for more details about this technique.

\newcommand{\bH}{\mathbb{H}}

\section{Algorithm for finding the optimal chain partition}\label{sect:alg}

Recall equation~(\ref{eq:F_pre}). This section shows how we compute $F[v]$ efficiently.

For computing $F[v]$, we use an interesting combined data structure, which we call \emph{heap-over-heap}.
Briefly, there are several disjoint heaps in the first layer (see $H_{x}^{(v)}$'s below),
  and the minimum elements of all heaps in the first layer are maintained in another heap (see $\bH^{(v)}$ below) at the second layer. 

\begin{description}
	\item [Min-heap $H_{x}^{(v)}$ for $x \in \smax_v$.] ~\\
            $H_{x}^{(v)}$ stores $i$ with key $\cost(x,i)$, if ($i=x$) or ($i \in \nmax_v$ and $\next_v(i)=x$). \smallskip
            
		\item [Min-heap $\bH^{(v)}$.] ~\\
             $\bH^{(v)}$ stores $i=\text{find-min}(H_{x}^{(v)})$ with key $\cost(v,i)$, for each $x \in \smax_v$. 
\end{description}

Each $i\in \win_v$ belongs to exactly one heap in $\{ H_{x}^{(v)}: x \in \smax_v\}$. See Figure~\ref{fig:heaps}.

\begin{remark}
Be aware that $i$ in $H_{x}^{(v)}$ is associated with a key value $\cost(x,i)$ instead of $\cost(v,i)$.
This is crucial to the entire algorithm.
  If $i$ is associated with $\cost(v,i)$ in the first layer, it is not easy to update the values when $v$ is changed.
\end{remark}

\begin{figure}[hb]
	\centering
	\includegraphics[width=\textwidth]{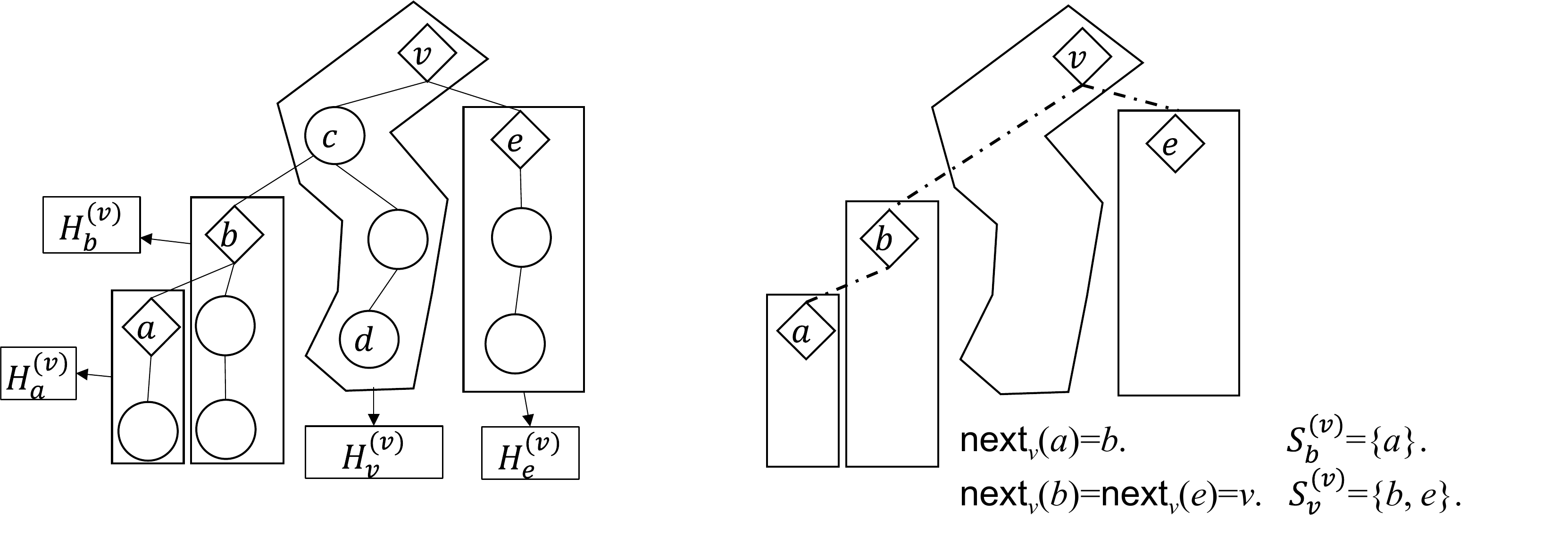}
	\caption{A partition of $\win_v$ into different parts, with each part organized by a heap.
        In this figure, rhombuses indicate s-maximal vertices and circles indicate other vertices.}\label{fig:heaps}
\end{figure}

\begin{lemma}\label{lem:findmin} 
$\arg\min_{i\in \win_v}\cost(v,i)=\findmin(\bH^{(v)})$.
In other words, computing $F[v]=\min_{i\in \win_v}\cost(v,i)$ reduces to computing the top element of $\bH^{(v)}$.
\end{lemma}

\begin{proof}

    Take an arbitrary $i \in \win_v$. 
    Assume $i$ is in $H_{x}^{(v)}$. 

    Let $i' = \findmin(H_{x}^{(v)})$. We have $\cost(x,i) \geq \cost(x,i')$.
	
	Since $i,i'$ are in $H_{x}^{(v)}$, respectively, applying Lemma~\ref{lem:cost substract}, 
	\begin{eqnarray*}
	   \cost(v,i) - \cost(x,i) &=& \cost(v,x)-\cost(x,x) \\
	   \cost(v,i') - \cost(x,i') &=& \cost(v,x)-\cost(x,x).
	\end{eqnarray*}
    Therefore, $\cost(v,i)-\cost(v,i')=\cost(x,i)-\cost(x,i')\geq 0$, namely, $\cost(v,i)\geq \cost(v,i')$.
	Since $a = \findmin(\bH^{(v)})$ and $i'$ is also in $\bH^{(v)}$, we have $\cost(v,i') \geq \cost(v,a)$.
    Thus for any $i \in \win_v$, $\cost(v,i) \geq \cost(v,a)$. 
	
	Further since $a\in \win_v$ (trivial), we have $a=\arg\min_{i\in \win_v}\cost(v,i)$.\qed
\end{proof}

For maintaining the heap $\bH^{(v)}$ over heaps $H_{x}^{(v)}$'s structure efficiently, we have to maintain two more types of heaps:
(1) A heap (see $W^{(v)}$ below), which is used for detecting the vertices leaving the window $\win_v$.
(2) A bunch of heaps (see $S_{x}^{(v)}$'s below), which are used for organizing the s-maximal vertices.

\begin{description}
\item [Max-heap $W^{(v)}$.]~\\ 
            $W^{(v)}$ stores $i$ with key $\sum_{j\in T[v,i]}w_j$, for each $i \in \win_v$.\smallskip

\item [Min-heap $S_{x}^{(v)}$ for $x \in \smax_v$.] ~\\
            $S_{x}^{(v)}$ stores $i$ with key $s_i$,  for each $i \in \smax_v$ with $\next_v(i)=x$. See Figure~\ref{fig:heaps}.
\end{description}

\begin{remark} 
In our algorithm, for each $x$,
  we do not have to store all the $H_x^{(v)}$'s one by one for each $v$ with $x\in \smax_v$.
Instead, all such $H_x^{(v)}$'s can share one entity in the memory, denoted as $H_x$.
The notation $H_x^{(v)}$ can be viewed as a historical version of $H_x$ at the stage of computing $v$.
The same logic holds for $S_x^{(v)}$.
\end{remark}

\begin{figure}[h]
	\centering
	\includegraphics[width=0.95\textwidth]{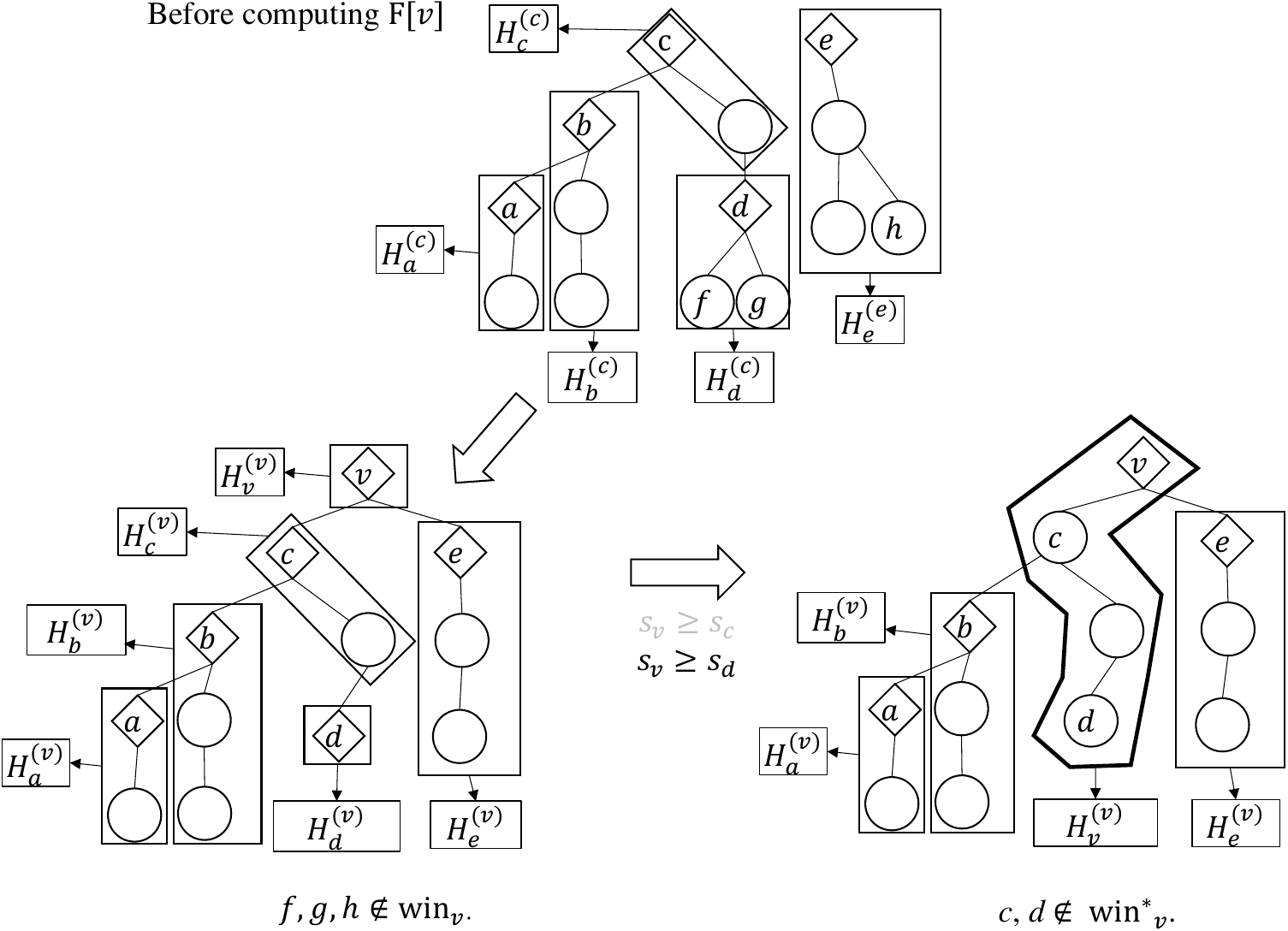}
	\caption{Illustration of the change of the heaps at the stage for computing $F[v]$.}\label{fig:alg_1}
\end{figure}

\subsection{Maintain the heaps efficiently}\label{sect:main}

We compute $\{F[v]\}$ from bottom to up. Suppose $F[v]$ is to be computed in the algorithm. 
It means that we have already computed the $F[c]$ for all $c\in \Children(v)$ and built all the aforementioned heaps related to $T_c$.
Our current task is to obtain the heaps related to $T_v$ from those heaps built for $T_c$'s. 

Consider an example as shown in Figure~\ref{fig:alg_1}. At the stage for computing $F[v]$, 
\begin{enumerate}
\item vertices $f,g,h$ are no longer in the window (namely, $f,g,h\in \win_c \setminus \win_v$),
        and hence should be eliminated from the heaps containing them;
\item vertices $c,d$ are no longer s-maximal (as $s_v\geq s_c$ and $s_v\geq s_d$),
        and hence $\{S_{x}^{(v)}: x\in \win^*_v\}$ changes accordingly;
        this leads to some merges of heaps in $\{H_{x}^{(v)}: x\in \win^*_v\}$, and 
           furthermore, $\bH^{(v)}$ should be updated accordingly.
\end{enumerate}

From the example, we see two major procedures are required at the moment: 
\begin{enumerate}
\item[\textbf{(a)}] Find $i\in \win_c\setminus \win_v$. Then, eliminate them from corresponding heaps.

\item[\textbf{(b)}] Find $x\in \win^*_c\setminus \win^*_v$ (it loses s-maximality due to $s_x\leq s_v$), and
            find those $i$ for which $\next_c(i)=x$ (due to $\next_v(i)=v\neq x$) (by utilizing $S_{x}^{(v)}$). 
          Then, modify $\{S_{x}^{(v)}: x\in \win^*_v\}$ and the heap-over-heap structure accordingly.
\end{enumerate}

 We present the main process of our algorithm as Algorithm~\ref{alg:compute_F(v)} below.

\begin{algorithm}[h]
	\caption{compute\_$F(v)$}\label{alg:compute_F(v)}\label{alg:main}
	\For {$c \in \Children(v)$}{
		$F[c] \leftarrow$ compute\_$F(c)$\;
	}
	$a\leftarrow \sum_{c\in \Children(v)}F[c]$\;
    $H_{v}^{(v)}$, $\bH^{(v)}$, $S_{v}^{(v)}$,$W^{(v)}$ $\leftarrow \emptyset$\;
	$\ins(H_{v}^{(v)},v,s_v+a)$\;\label{line:ins_i_Hfv}
	\For {$c \in \Children(v)$}{\label{line:1:Cv-s}
		$\ins(S_{v}^{(v)},c,s_c)$\;\label{line:ins_i_Sv}
		$\addall(W^{(v)},w_v)$\;
		Over-weight($c$)\;\label{line:Process_Over-weight()}
		$\meld(W^{(v)},W^{(v)})$\;\label{line:meld_W_i}
		$b\leftarrow a-F[c]$\;
		$\addall(\bH^{(c)},b)$\;\label{line:update_sum_m_s}
		$\meld(\bH^{(v)},\bH^{(c)})$\;\label{line:1:Cv-t}
	}
	Update-s-maximal($v$)\;\label{line:Process_Update-s-maximal()}
	$\ins(W^{(v)},v,w_v)$\;\label{line:ins_w_v}
	$\ins(\bH^{(v)},\findmin(H_{v}^{(v)}),\keyvalue(H_{v}^{(v)},\findmin(H_{v}^{(v)})))$\;
	\Return $\keyvalue(\bH^{(v)},\findmin(\bH^{(v)}))$\;
\end{algorithm}

In the algorithm, we first add $v$ into $H_{v}^{(v)}$ with key $s_v + \sum_{c\in \Children(v) }F[c]$ for initialization of $v$. Then, for every $c\in \Children(v)$, we do the following:

\begin{itemize}
	\item First, add $c$ to $S_{v}^{(v)}$ as part of the initialization of $S_{v}^{(v)}$, since $c$ is s-maximal by default in $T_c$. We assume for the time being $c$ is still s-maximal in $T_v$ before we check $s_i$ in the following process. Therefore, $c$ first goes to $S_{v}^{(v)}$.
	
\smallskip
	\item We then add $w_v$ to all keys in $W^{(v)}$, so that the vertices in $W^{(v)}$ with key greater than $w_0$ are vertices in $\win_c\setminus\win_v$ which need to be deleted (recall procedure (a)).
    We call a function Over-weight($c$) to delete the corresponding vertices from all heaps; see subsection~\ref{sect:OB} for details.
    After Over-weight($c$), we meld $W^{(c)}$ into $W^{(v)}$.

\smallskip
	\item Next, we meld $\bH^{(c)}$ into $\bH^{(v)}$ as a part of the initialization of $\bH^{(v)}$. By Lemma~\ref{lem:cost substract}, since $c$ is considered to be s-maximal in $T_c$ by default, we add $(\sum_{c\in \Children(v)}F[c])-F[i]$ to all vertices in $\bH^{(c)}$ to update the $\chainsum(c,*)$ of every key to $\chainsum(v,*)$ before melding. As for updating the $s_i$ part of the costs, we leave this issue to function Update-s-maximal($v$) later.
\end{itemize}

After the above process for all $c\in\Children(v)$, heap $S_{v}^{(v)}$ contains all $\Children(v) \cap \win_v$. $W^{(v)}$ and all $H$'s only contain vertices in $\win_v$, which are the only vertices that we need for the heaps of $T_v$.

We now focus on the whole $T_v$. Here, we consider the vertices that lose its s-maximal identity after adding $v$ as the new root (as mentioned in procedure (b)). We process these vertices in Update-s-maximal($v$) with consecutive updating and melding (See details in appendix~\ref{sect:s-max}). 

After Update-s-maximal($v$), we add $v$ into $W^{(v)}$ with key $w_v$ for further recursive calls and insert the top vertex of new $H_{v}^{(v)}$ into $\bH^{(v)}$. At this point, we have all heaps meeting their definitions in $T_v$ (see subsection~\ref{sect:correct} for details), and all we need is to extract minimum choice from $\bH^{(v)}$ to obtain $F[v]$.

\subsection{Analysis}



\begin{lemma}\label{lemma:time_Pm}
	\textnormal{Over-weight()} takes $O(n\log n)$ time in total.
\end{lemma}

\begin{proof}
	See appendix~\ref{sect:OB} along with the description of Over-weight().
\end{proof}

\begin{lemma}\label{lemma:time_Ps}
	\textnormal{Update-s-maximal()} takes $O(n\log n)$ time in total.
\end{lemma}

\begin{proof}
	See appendix~\ref{sect:s-max} along with the description of Update-s-maximal().
\end{proof}

\begin{theorem}\label{thm:alg}
	The sum-of-max chain partition problem can be solved in $O(n \log n)$ time with our algorithm.
\end{theorem}

\begin{proof}
	By Lemma~\ref{lemma:time_Pm} and Lemma~\ref{lemma:time_Ps}, the functions Over-weight() and Update-s-maximal() require $O(n \log n)$ time in total. 
	
	In the main process of each $v$,  we insert $v$ into $H_{v}^{(v)}$ with \ins() in $O(\log n)$ time. Note that the acquisitions of $a=\sum_{c\in \Children(v) }F[c]$ in total take time
	$$O((\sum_{v \in T}|\Children(v)|)\log n)=O(n \log n).$$
	
	Befrore Update-s-maximal($v$), for each $c\in\Children(v)$, we insert $c$ into $S_{v}^{(v)}$ with \ins() in $O(\log n)$ time. We then update $W^{(v)}$ in $O(\log n)$ time by \addall(). After Over-weight($c$), we also update and meld $\bH^{(c)}$ into $\bH^{(v)}$ with \addall() and \meld() in $O(\log n)$ time. Here, we conduct $O(1)$ times of $O(\log n)$ operations (excluding Over-weight($c$)) for every $c\in \Children(v)$ as described above. These operations in total also take time $$O((\sum_{v \in T}|\Children(v)|)\log n)=O(n \log n).$$
	
	After Update-s-maximal($v$), we insert $v$ into $W^{(v)}$ and insert the top vertex of the new $H_{v}^{(v)}$ into $\bH^{(v)}$ with \ins() in $O(\log n)$ time. The extraction of $F[v]$ is also $O(\log n)$ with \findmin() and \keyvalue(). Here, we conduct $O(1)$ times of $O(\log n)$ operations (excluding Update-s-maximal($v$)) for $v$, which is also $O(n \log n)$ in total.
	
	Therefore, during the whole recursive process, the algorithm takes $O(n \log n)$ time.\qed
\end{proof}

\begin{remark}\label{rem:next}
The maintenance and acquisition of $\next_*()$ are $O(n\alpha(n))< O(n\log n)$ in total. Therefore, function $\next_*()$ does not affect the $O(n \log n)$ time bound.

	The acquisition of $\next_*()$ can be done in $O(\alpha(n))< O(\log n)$ time with the find() operation of disjoint-set data structures. To achieve that, all we need to do is to set every vertex $\next_c(c)$ of $c \in \Children(v)$ as $v$ when $F[v]$ starts to be computed and merge the sets of vertex $x\in T_v$ and $v$ when $x$ loses its s-maximal identity in Update-s-maximal($v$). 
	
	As a result, the find() operation can be seen as $\next_c()$ before Update-s-maximal($v$) and $\next_v()$ after Update-s-maximal($v$) by the definition of $\next$. Thus we have $\next_c()$ and $\next_v()$ ready for every $T_c$ of $c\in \Children(v)$ in the main process of computing $F[v]$.
	
\end{remark}

\section{A generalized chain partition problem}\label{sect:generalize}
This section discusses a generalization of our algorithm.

Suppose every vertex $a$ of the tree is associated with a distinct rank $r_a$.
The \emph{dominant vertex} of a chain refers to the one in the chain with the largest rank.
	
	\begin{quote}\label{problem 2}
		\textbf{A chain partition problem.} 
		Given a rooted tree $T$ of $n$ nodes, where vertex $i$ is associated with a weight $w_i$, a cost $s_i$, and a rank $r_i$. Partition the tree into several disjoint chains $C_1,\ldots,C_k$, where $k$ is arbitrary, so that the total weight of each chain is no more than a given threshold $w_0$ and the sum of the costs of the chains is minimized, where the cost of a chain is given by $s_i$ -- where $i$ denotes the dominant vertex of the chain. 
	\end{quote}

This problem contains the previous problem as a special case; by setting $r_i=s_i$ we obtain the original sum-of-max problem.

\medskip To tackle this problem, we replace the concept s-maximal with r-maximal, and this also leads to a change of the definition of $\next_v()$ and $\smax$:
	
	\begin{definition}\label{def:gen}
		A vertex $u$ in $T_v$ is \emph{r-maximal} if each vertex $p$ in $T[v,u)$ admits that $r_p<r_u$. 

		For $u\in T_v$ and $u\neq v$, let $\next_v(u)$ be the closest r-maximal proper ancestor of $u$ in $T_v$.

		Let $\smax_v$ denote the r-maximal vertices in $\win_v$ and $\nmax_v$ denote the vertices in $\win_v\setminus\smax_v$.
	\end{definition}	 

	We now slightly modify our algorithm to solve the generalized problem:	

	One thing we have to change lies in the definition of heaps:
      $S_{x}^{(v)}$'s keys should be the $r_i$'s of vertices instead of $s_i$'s. 
	Another change we make lies in the function Update-s-maximal() in subsection~\ref{sect:s-max}. 
       All we need to do is to change the criteria from s-maximal to r-maximal (by changing the condition $\leq s_v$ to $\leq r_v$).
	
	The definitions above still fit Lemma~\ref{lem:findmin}, equation~(\ref{eq:F_pre}) and (\ref{eq:cost}) and all relevant analysis in this paper. 
     The heaps can be maintained and implemented to compute $F[v]$'s using the algorithm described above. As a summary, we conclude that
	
	\begin{theorem}\label{thm:generalization}
		The chain partition problem given above can be solved in $O(n \log n)$ time.
	\end{theorem}

	
\section{Conclusions}
In this paper, we present an $O(n \log n)$ algorithm for the sum-of-max chain partition of a tree under a knapsack constraint. 
The algorithm is based on non-trivial observations (such as \Cref{lem:cost substract,lem:findmin}).
The key ingredient of our algorithm is a delicately designed data structure that consists of two layers of heaps, 
where the elements of heaps in the second layer $\bH^{(v)}$ come from the minimum elements of the heaps in the first layer $H^{(v)}$'s. 
We also generalize the algorithm to handle a more complex scenario (\Cref{sect:generalize}), and the time complexity remains $O(n \log n)$.

\smallskip 
One question worth exploring in the future is how to find the sum-of-max path partition of a tree.
Moreover, can we find the optimal chain partition of a tree when the cost of a chain $C$ is defined as $\max_{i\in C} s_i - \min_{i\in C} s_i$?


%
%
%
\bibliography{ref}
\bibliographystyle{splncs04}

\clearpage
\appendix
\section{Binomial Heap with Lazy Tags}\label{app:lazy_tag}

For each node in the binomial heap, we add a lazy tag to it, representing adding or subtracting a value uniformly for all elements below. In this section, we briefly explain the differences brought by the lazy tag technique.

Traditionally, when using lazy tags, we pushdown a tag to all children when the node is accessed. However, this would lead to a non-$O(1)$ pushdown operation in binomial heaps since there can be $O(\log n)$ children for one node, which may cause an $O(\log^2 n)$ factor in time bound. 

Therefore, we do not pushdown a lazy tag unless it's a tag of a root vertex to be deleted. This change requires us to handle the tags in a slightly different way. First, we only exert a negative pushdown to the newly melded tree as shown in Figure~\ref{fig:tag-meld}. Suppose tree $B_a$ with a bigger root is linked to $B_b$, then the opposite number of the lazy tag of $B_b$ will is added to the lazy tag of $B_a$ during the linking operation. This trick allows us we to preserve the property of the heap while the negative pushdown is still an $O(1)$ operation.

\begin{figure}[htbp]\label{fig:tag-meld}
	\centering
	\includegraphics[width=0.75\textwidth]{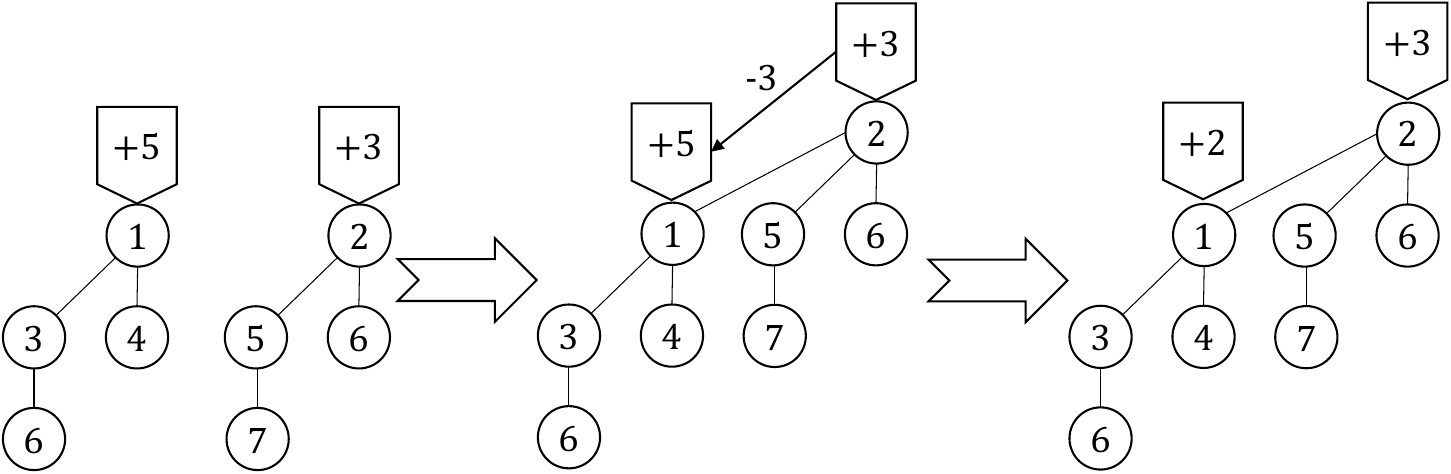}
	\caption{Example of binomial heaps melding with lazy tag}\label{fig:meld}
\end{figure}

Besides, when pushing up a node for a key value change. The key value of the node swapped down needs to be changed since the tag stays still. If a vertex is swapped down, it needs to counter the effect caused by the lower tag and thus add the opposite of the lower tag, see Figure~\ref{fig:pushup} for illustration.

\begin{figure}[htbp]
	\centering
	\includegraphics[width=0.8\textwidth]{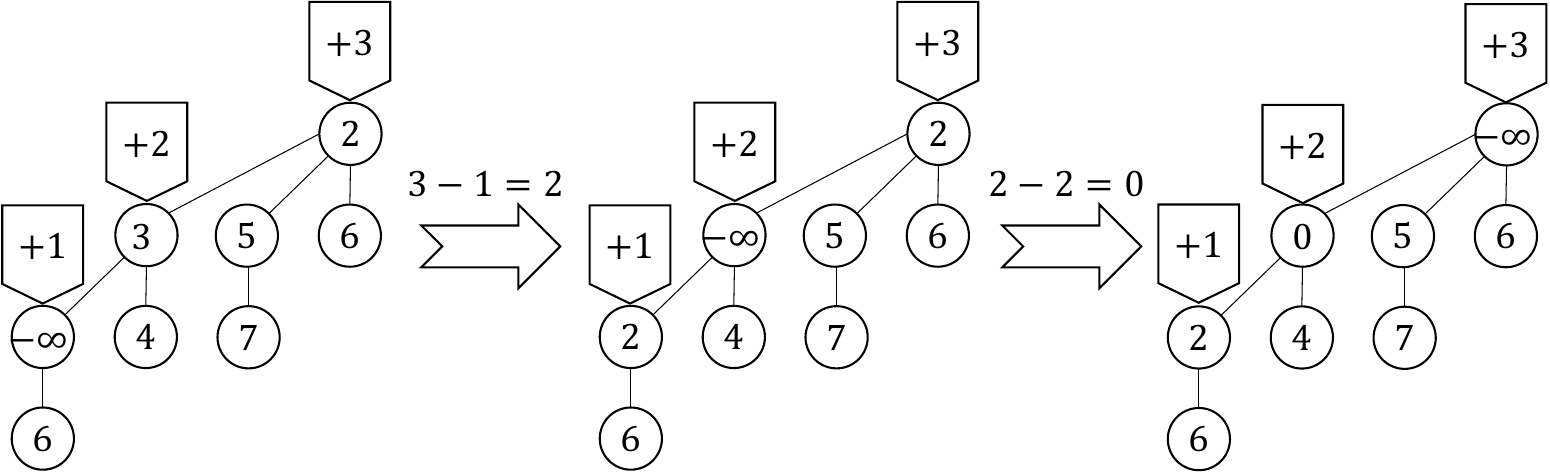}
	\caption{Example of binomial heaps pushing up a node $-\infty$. In the first step, 3 is swapped down and becomes $3-1=2$ to counter the $+1$ tag's effect since it was not affected by $+1$ in the first place. Similarly, 2 is swapped down and becomes $2-2=0$ to counter the $+2$ tag.}\label{fig:pushup}
\end{figure}

Thanks to the $O(\log n)$ depth limit of binomial heap, even with lazy tags affecting the real key value of each node, we can still also acquire a node's key-value in $O(\log n)$ time by simple tree-climbing towards the root.

\section{Details of the functions}\label{app:func}

Before we analyze the functions, we first present an lemma. The lemma and its corollaries strictly limits the number of calls of the functions in this section, contributing a key ingredient of our analysis.

\begin{lemma}\label{lem:once}
	Given tree $T$, for a pair of vertex $u, v\in T$, if $u$ is an ancestor of $v$ we have
	\begin{itemize}
		\item[(i)] if vertex $i$ is a descendant of $v$, $i\notin \win_v$, then $i\notin \win_u$.
		\item[(ii)] if vertex $i$ is a descendant of $v$, $i\notin \smax_v$, then $i\notin \smax_u$.
	\end{itemize}
	As a corollary of (i), during our recursive process from the bottom up, if a vertex $i$ leaves the window of $T_v$, then we shall never have it in our window again when processing $v$'s ancestors in further procedures.
	
	And as a corollary of (ii), during our recursive process from the bottom up, if a vertex $i$ lost its s-maximal identity in $T_v$, then we shall never have it s-maximal again when processing $v$'s ancestors in further procedures.
\end{lemma}

\begin{proof}[of (i)]
	Suppose $i$ is a descendant of $v$, $i\notin \win_v$. By definition, we have $\sum_{j\in T[v,i]}w_j> w_0$. Since $u$ is the ancestor of $v$, we have $$\sum_{j\in T[u,i]}w_j = \sum_{j\in T[u,v)}w_j + \sum_{j\in T[v,i]}w_j> \sum_{j\in T[v,i]}w_j> w_0.$$
	The inequality above indicates that $i\notin \win_u$. \qed
\end{proof}

\begin{proof}[of (ii)]
	Suppose $i$ is a descendant of $v$, $i\notin \smax_v$. We may assume that $i\in \win_u$, otherwise $i\notin \smax_u$ trivially.
	
	By definition, we know that there exists a vertex $p\in T[v,i)$ admitting $s_i<s_p$. Since $T[v,i)\subset T[u,i)$, we also have $p\in T[u,i)$. Thus $i$ is also not s-maximal in $T_u$, and therefore $i\notin \smax_u$.\qed
\end{proof}

\subsection{Critical function: Over-weight()}\label{sect:OB}

In Over-weight($c$), we process vertices $x\in \win_c\setminus \win_v$ (exiting the window), as shown in Figure~\ref{fig:alg_1}. 

Since we temporarily consider $c$ to be s-maximal here as the root of $T_c$, $\next_c()$ and $\next_v()$ are essentially the same. We only use $\next_c()$ here and only change it to $\next_v()$ permanently when executing Update-s-maximal($v$) later.

At the beginning of Over-weight($c$), we search for vertices in $W^{(v)}$ with keys larger than $w_0$ in $O(\log n)$ time each with $\findmin()$. These vertices are the ones leaving the wind by definition with $\sum_{j\in T[v,*]}w_j > w_0$. During the function, we delete them from the heaps from the bottom up. Supposing a vertex $x \in T_c$ is such vertex, we need to:

\textbf{(i)} Delete $x$ from $H_{\next_c(x)}^{(v)}$ and $S_{\next_c(x)}^{(v)}$. 

\textbf{(ii)} Replace $x$ in $\bH^{(c)}$ with the next smallest vertex in $H_{\next_c(x)}^{(v)}$ (if exists), since $x$ is not in $\win_v$ anymore. 

\textbf{(iii)} Delete $x$ from $W^{(v)}$ since it's not in $\win_v$ and thus not needed to be melded into $W^{(v)}$. 

In (i), we complete the deletion with $\deletevertex()$. 

In (ii), we also need to replace $x$ with new top element $k$ of $H_{\next_c(x)}^{(v)}$. By Lemma~\ref{lem:cost substract}, we compute $\cost(c,k)$ as follows.
\begin{equation}
	\cost(c,k)=\cost(\next_c(x),k) + \cost(c,x)-\cost(\next_c(x),x).
\end{equation}
Since $\cost(c,x)$ is in $\bH^{(c)}$, $\cost(\next_c(x),x)$ and $\cost(\next_c(x),k)$ are in  $H_{\next_c(x)}^{(v)}$, we can obtain these values with $\keyvalue()$.

In (iii), we simply delete $x$ from $W^{(v)}$.

After Over-weight($c$) for every $c\in\Children(v)$, vertices in $\win_v$ are the only vertices kept in the heaps. See Algorithm~\ref{alg:weight} in Appendix~\ref{app:alg} for pseudocode of Over-weight($c$).

\begin{proof}[of Lemma~\ref{lemma:time_Pm}]
	For every vertex $x$ in $T_c$ leaving $\win_v$, it takes $O(1)$ time for us to search it from binomial heap $W^{(v)}$ with \findmin(). Then it cost us $O(\log n)$ time to delete it from $H_{\next_c(x)}^{(v)}$ and potential $\bH^{(c)}$. It costs us $O(\log n)$ time to find $k$ and insert it into $\bH^{(c)}$ as a replacement with \findmin() and \ins(). We also need $O(\log n)$ time to delete it from its previous $S_{\next_c(x)}^{(v)}$ with \deletevertex().
	
	By Lemma~\ref{lem:once}, every vertex only exit the window once with $O(1)$ number of $O(\log n)$ time operations. As a result, Over-weight() takes $O(n\log n)$ time in total.\qed
\end{proof}

\subsection{Critical function: Update-s-maximal()}\label{sect:s-max}

After Over-weight($c$) for every $c\in\Children(v)$, all vertices $\win_v$ are kept in the heaps. But still, for $c\in \Children(v)$ and its $T_c$, some vertices in $\smax_c$ may not be in $\smax_v$ for not being s-maximal anymore, we mainly deal with these vertices in Update-s-maximal().

For vertex $c\in \Children(v)$, if a vertex $x$ is in $\smax_c$ but not in $\smax_v$, it means its s-maximal identity is lost in $T_v$. In fact, we have the following lemma stating that $x\in\smax_c$'s s-maximal identity can only be lost to $v$ in $T_v$ and $v$ will inherit vertices $i\in \win_v$ with $\next_c(i)=x$.

\begin{lemma}\label{lem:s_v geq s_x}
	If $x\in\smax_c\setminus\smax_v$ for a $c\in \Children(v)$, then we have:
	(i)$s_v\geq s_x$ and (ii)$\next_v(i)=v$ for all $i\in \win_v$ with $\next_c(i)=x$.
\end{lemma}

\begin{proof}
	By the definition of $\smax_c$, we know that $s_x>s_p$ for all $p\in T[c,x]$. Since $x\notin\smax_v$, $\exists q\in T[v,x)$ where $s_q\geq s_x$. As a result,  we have (i):$s_v\geq s_x$. 
	
	As for (ii), since $\next_c(i)=x$, we have $s_x\geq s_p$ for all $p\in T[c,i]$. Therefore, we have $s_v\geq s_p$ for all $p\in T[v,i]$ and thus only $v$ is s-maximal in $T[v,i]$ for $T_v$, $\next_v(i)=v$.\qed
\end{proof}

As a result, we can search for such $x$ with heap $S_{v}^{(v)}$ and process them accordingly.

\medskip In the main process, firstly, all $c\in \Children(v)$ are stored in $S_{v}^{(v)}$. For a certain $x\in S_{v}^{(v)}$, if $s_x > s_v$, then no vertex in $T_x$ would be affected by $s_v$. Otherwise, $x$ loses its s-maximal identity and any vertex $i\in\win_v$ with $\next_c(i)=x$ would have $\next_v(i)=v$. Thus, by definition, $H_{x}^{(v)}$ and $S_{x}^{(v)}$ should be melded into $H_{v}^{(v)}$ and  $S_{v}^{(v)}$ respectively. This can be done in $O(\log n)$ time with \meld().

\medskip Recall the definition of $H_{v}^{(v)}$. This melding process needs updating the costs in $H_{x}^{(v)}$ from $\cost(x,*)$ to $\cost(v,*)$ first. The update is as follows.

For all vertices $i \in H_{x}^{(v)}$, by Lemma~\ref{lem:cost substract}, $g = \cost(v,i)-\cost(x,i)$ is a fixed number for fixed $x$ and $v$. Therefore, we only need to add $g$ to their keys to update its cost with \addall() in $O(\log n)$ time to complete the update. Now we discuss the computation of $g$.

Suppose $z \in H_{x}^{(v)}$ is the minimum element of $H_{x}^{(v)}$. Recall the main process, the key value of $z$ in $\bH^{(v)}$ should be $(\chainsum(v,z)+s_x)=(\cost(v,z)+s_x-s_v)$ now. This is because its $\chainsum(c,z)$ has been updated to $\chainsum(v,z)$ in the main process before melding $\bH^{(c)}$ into $\bH^{(v)}$. However, we have assumed $x$ to be s-maximal in that part, leaving $s_x$ of the key value not updated to $s_v$.

Having $(\cost(v,z)+s_x-s_v)$ in $\bH^{(v)}$ and $\cost(x,z)$ in $H_{x}^{(v)}$, we can acquire $(\cost(v,z)+s_x-s_v)$ and $\cost(x,z)$ in $O(\log n)$ time. After that, we can compute
\begin{equation}
	\begin{aligned}
		g &= (\cost(v,z)+s_x-s_v)-\cost(x,z) +s_v-s_x\\
		&= \cost(v,z)-\cost(x,z)\\
		&= \cost(v,i)-\cost(x,i).
	\end{aligned}
\end{equation}

So for all $i \in H_{x}^{(v)}$, we uniformly add $g = (\cost(v,i)-\cost(x,i))$ to them to update their cost from $\cost(x,i)$ to $\cost(v,i)$ with \addall(). Then, we meld $H_{x}^{(v)}$ with updated cost values into $H_{v}^{(v)}$. As for $x$, since it is no longer s-maximal, we delete it from $S_{v}^{(v)}$. Besides, when $x$ lost its s-maximal identity, its corresponding heaps' corresponding vertex in $\bH^{(v)}$ are also deleted without replacement, since the heap itself is melded already.

\medskip After melding $H_{x}^{(v)}$ into $H_{v}^{(v)}$ and $S_{x}^{(v)}$ into $S_{v}^{(v)}$ respectively. We continue to search for the next smallest element in $S_{v}^{(v)}$ and repeat the previous process until no element in $S_{v}^{(v)}$ has a key smaller than $s_v$.

\begin{remark}
	Note that $x$ may not be in $S_{v}^{(v)}$ in the first place, but if $x$ lost its s-maximal identity, then $\next_c(x)$ also loses its s-maximal identity since $s_v\geq s_x>s_{\next_c(x)}$. So $S_{v}^{(v)}$ first detects and processes $\next_c(x)$ before $x$ and detects $x$ after $S_{\next_c(x)}^{(v)}$ is melded into $S_{v}^{(v)}$. This process is transitive.
\end{remark}

\begin{remark}
	From here, the reason why we design two layers of heaps is obvious. This is because $H_{x}^{(v)}$'s serving as the first layer, are defined based on $\next_v(i) = x$. With the heaps accurately locating the vertices whose costs contain $s_x$, we can now update them uniformly with an $\addall()$ operation in $O(\log n)$ time when the $s_x$ is replaced by $s_v$.
\end{remark}

See Algorithm~\ref{alg:phase 2} in Appendix~\ref{app:alg} for pseudocode of Update-s-maximal($v$).

\begin{proof}[of Lemma~\ref{lemma:time_Ps}]
	In Update-s-maximal(), when a vertex $x$ is not s-maximal anymore, we meld its corresponding heaps $H_{x}^{(v)}$,  and $S_{x}^{(v)}$ into $H_{v}^{(v)}$ and $S_{v}^{(v)}$ respectively in $O(\log n)$ time with \meld(). For the melded $H_{x}^{(v)}$, we also delete its top element from $\bH^{(v)}$ respectively in $O(\log n)$ time with \deletevertex(). For $x$, since its no longer s-maximal, we also delete it from $S_{v}^{(v)}$ in $O(\log n)$ time with \deletevertex(). 
	
	It takes $O(n \log n)$ time to deal with such cases in total, because by Lemma~\ref{lem:once}, a vertex can only lost its s-maximal identity once during the recursive process with $O(1)$ number of $O(\log n)$ time operations above. \qed
\end{proof}

\section{Extra discussion}\label{sect:correct}
In this section, we discuss more details of the correctness of the algorithm. Namely the correctness of the way we maintain the heaps. We demonstrate this heap by heap with four propositions.

\begin{proposition}\label{prop:cor-W}
	When extracting $F[v]$, we have $W^{(v)}$ storing $i\in \win_v$ with key $\sum_{j\in T[v,i]}w_j$. All nodes leaving $\win_v$ are deleted from all heaps.
\end{proposition}

\begin{proof}
	 $\win_v \subseteq (\bigcup_{c\in \Children(v)}\win_c\cup\{v\})$ and we have deleted all $(\bigcup_{c\in \Children(v)}\win_c\cup\{v\})\setminus \win_v$ from all heaps in Over-weight($c$), so we have exactly the vertices in $\win_v$ stored in $W^{(v)}$ and all vertices leaving $\win_v$ deleted from all heaps. The correctness of the key values of $W^{(v)}$ follows from the update of $W^{(c)}$  before melding it into $W^{(v)}$.\qed
\end{proof}

\begin{proposition}\label{prop:cor-S}
	When extracting $F[v]$, we have $S_{x}^{(v)}$ storing $i$ with key $s_i$ for each $i \in \smax_v$ and $\next_v(i)=x$.
\end{proposition}

\begin{proof}
	Note that before the main process, $S_{x}^{(v)}$ stores $i$ with key $s_i$  if $i \in \smax_c$ and $\next_c(i)=x$ for every $c\in \Children(v)$, which are the nodes in $\bigcup_{c\in \Children(v)}\smax_c$ except all $c\in \Children(v)$. Since we have inserted $c\in \Children(v)$ into $S_v^{(v)}$ in the main process, if a vertex needs to be in an $S_{x}^{(v)}$ in $T_v$, it is in an $S_{x}^{(v)}$ now. To show that every vertex is in its corresponding $S_{x}^{(v)}$ properly, we discuss three cases:
	
	(i)Suppose $i \in \smax_v \cap T_c$ and $\next_v(i)=x=\next_c(i)$ for a $c\in \Children(v)$. Then we have $x \in \smax_c\cap\smax_v$. We also have $s_x>s_v$ since $x \in \smax_v$, so $S_{x}^{(v)}$ is not affected throughout the algorithm (Except trivially deleting the vertices exiting the window). Then we have all $i \in \smax_v$ and $\next_v(i)=\next_c(i)=x$ stored in $S_{x}^{(v)}$ when extracting $F[v]$. Such $i$'s are properly stored in their corresponding $S_{x}^{(v)}$.
	
	(ii)Suppose $i \in \smax_v \cap T_c$ and $\next_v(i)\neq\next_c(i)=x$ for a $c\in \Children(v)$. By Lemma~\ref{lem:once}, no vertex in $T(x,i]$ can become s-maximal in $T_v$ since they are not s-maximal in $T_c$ already. Then $x$ must have lost its s-maximal identity to $v$ in $T_v$ by Lemma~\ref{lem:s_v geq s_x}. In this case, $S_{x}^{(v)}$ is melded into $S_{v}^{(v)}$ in Update-Update-s-maximal($v$) and $x$ is deleted from $S_{v}^{(v)}$. Then we have all $i \in \smax_v$ and $\next_v(i)\neq\next_c(i)=x$ properly stored in $S_{v}^{(v)}$ when extracting $F[v]$, where $\next_v(i)=v$ by Lemma~\ref{lem:s_v geq s_x}.
	
	(iii)Suppose $i \notin \smax_v$, then it is either not in an $S_{x}^{(v)}$ in the first place and thus will never be in one or it is in a $S_{x}^{(v)}$ at the beginning but got deleted since it have lost its s-maximal identity. In any case, it will not be in an $S_{x}^{(v)}$ when extracting $F[v]$.
	
	The correctness of the key values in $S_{x}^{(v)}$'s is trivial since the key value of vertex $i$ is $s_i$ all along. Thus we have $S_{x}^{(v)}$ storing $i$ with key $s_i$ for each $i \in \smax_v$ and $\next_v(i)=x$ before extracting $F[v]$.\qed
\end{proof}

\begin{proposition}\label{prop:cor-H}
	When extracting $F[v]$, we have $H_{x}^{(v)}$ storing $i$ with key $\cost(x,i)$ if ($i=x$) or ($i \in \nmax_v$ and $\next_v(i)=x$).
\end{proposition}

\begin{proof}
	First of all, we have exactly every $i\in\win_v$ stored in a certain $H_{x}^{(v)}$ before extracting $F[v]$. This is because $\win_v \subseteq \bigcup_{c\in \Children(v)}\win_c\cup\{v\}$, $v$ is inserted into $H_{v}^{(v)}$ and all nodes leaving $\win_v$ are deleted from all heaps by Proposition~\ref{prop:cor-W}.
	
	Similar to the proof of Proposition~\ref{prop:cor-S}, we can derive that: 
	
	(i)Suppose ($i \in \nmax_v \cap T_c$ and $\next_v(i)=\next_c(i)=x$) or ($i=x$ and $x\in\smax_c\cup\smax_v$) for a $c\in \Children(v)$. $H_{x}^{(v)}$ is not affected by the algorithm. Such $i$'s are properly stored in $H_{x}^{(v)}$.
	
	(ii)Suppose ($i \in \nmax_v \cap T_c$ and $\next_v(i)\neq\next_c(i)=x$)  or ($i=x$ and $x\in\smax_c\setminus\smax_v$) for a $c\in \Children(v)$. By Lemma~\ref{lem:s_v geq s_x}, $\next_v(i) = v$ and $s_v\geq s_x$, so we have $H_{x}^{(v)}$ melded into $H_{v}^{(v)}$ in Update-Update-s-maximal($v$). In this case, we have $\next_v(i) = v$ and $i$ is properly stored in $H_{v}^{(v)}$ before extracting $F[v]$.
	
	(iii)Suppose $i\notin \win_v$, $i$ is deleted in a call of Over-weight().
	
	The correctness of the key values follows from the updating procedure. In case (i), we the key values remain the same and thus are correct. In case (ii), the key values are updated to $\cost(v,i)$ according the description of Update-Update-s-maximal(). Thus, we have $H_{x}^{(v)}$ storing $i$ with key $\cost(x,i)$ if ($i=x$) or ($i \in \nmax_v$ and $\next_v(i)=x$).
\end{proof}

\begin{proposition}\label{prop:cor-bH}
	When extracting $F[v]$, we have $\bH^{(v)}$ storing $i=\text{find-min}(H_{x}^{(v)})$ with key $\cost(v,i)$, for each $x \in \smax_v$. 
\end{proposition}

\begin{proof}
	 We discuss all $H_{x}^{(v)}$'s in three cases:
	
	(i)For each $x \in \smax_v\setminus\{v\}$, we have $\bH^{(c)}$ storing $i=\text{find-min}(H_{x}^{(v)})$ with key $\cost(c,i)$ for a $c\in \Children(v)$ before the algorithm. Since the key value of $i$ has been updated to $\cost(v,i)$ and $i$ is melded into $\bH^{(v)}$ in the main algorithm, we have such $i$ properly stored in $\bH^{(v)}$.
	
	(ii)For $x = v$, we insert $i=\text{find-min}(H_{v}^{(v)})$ into $\bH^{(c)}$ with key $\cost(v,i)$ directly in the main process.
	
	(iii)For $x \in \smax_c \setminus \smax_v$, we delete $i=\text{find-min}(H_{x}^{(v)})$ from $H_{v}^{(v)}$ when $H_{x}^{(v)}$ is melded into $H_{v}^{(v)}$ in Update-s-maximal($v$).
	
	As a result, we have $\bH^{(v)}$ storing $i=\text{find-min}(H_{x}^{(v)})$ with key $\cost(v,i)$, for each $x \in \smax_v$. 
\end{proof}

To sum up, we have all heaps meeting their definitions for $T_v$ when extracting $F[v]$, the correctness follows.

\section{Pseudocode for the algorithm}\label{app:alg}

\begin{algorithm}[h]
	\caption{Function Over-weight($c$). Line~\ref{line:(1:s)},\ref{line:(1:t1)} and \ref{line:(1:t2)} are procedure (i). From line~\ref{line:(2:s)} to~\ref{line:(2:t)} are (ii). Line~\ref{line:(4)} is (iii).}\label{alg:weight}
	\While{$\keyvalue(W^{(v)}$,$\findmax(W^{(v)}))>w_0$}{\label{line:search >w0}
		$x \leftarrow \findmax(W^{(v)})$\; 
		$\deletevertex(S_{\next_c(x)}^{(v)},x)$\;\label{line:(1:s)}
		\If{$\findmin(\bH^{(c)}) = x$}{\label{line:(2:s)}
			$g \leftarrow \keyvalue(\bH^{(c)},x)-\keyvalue(H_{\next_c(x)}^{(v)},x)$\;\label{line:(1:t)}
			$\deletevertex(H_{\next_c(x)}^{(v)},x)$\;\label{line:(1:t1)}
			$k \leftarrow \findmin(H_{\next_c(x)}^{(v)})$\;
			$\ins(\bH^{(c)},k,\keyvalue(H_{\next_c(x)}^{(v)},k)+g)$\;\label{line:cost(i,k_2)}
			$\deletevertex(\bH^{(c)},x)$\;\label{line:(2:t)}
		}
		\Else{
			$\deletevertex(H_{\next_c(x)}^{(v)},x)$\;\label{line:(1:t2)}
		}
		$\deletevertex(W^{(v)},x)$\;\label{line:(4)}
	}
\end{algorithm}

\begin{algorithm}[h]
	\caption{Update-s-maximal($v$)}\label{alg:phase 2}
	\While{\keyvalue($S_{v}^{(v)}$,$\findmin(S_{v}^{(v)}))\leq s_v$}{
		$x \leftarrow \findmin(S_{v}^{(v)})$\;

		$z \leftarrow \findmin(H_{x}^{(v)})$\;
		$g \leftarrow \keyvalue(\bH^{(v)},z)-\keyvalue(H_{x}^{(v)},z)+s_v-s_x$\;
		$\addall(H_{x}^{(v)},g)$\;
		$\meld(H_{v}^{(v)},H_{x}^{(v)})$\;
		$\deletevertex(\bH^{(v)},z)$\;

		$\meld(S_{v}^{(v)},S_{x}^{(v)})$\;
		$\deletevertex(S_{v}^{(v)},x)$\;	
	}
\end{algorithm}

\end{document}